\documentclass[10pt,a4paper]{article}

\usepackage{amsmath,amsgen,latexsym}
\usepackage{amstext,amssymb,amsfonts,latexsym}
\usepackage{theorem}
\usepackage{pifont}
\usepackage[hyphens]{url}
\usepackage{microtype}
\usepackage{graphicx}

 \newcommand{\bs}{\bigskip}
 
 \newcommand{\n}{\noindent}
 \newcommand{\s}{\smallskip}
 \newcommand{\hs}[1]{\hspace*{ #1 mm}}
 \newcommand{\vs}[1]{\vspace*{ #1 mm}}

 \newcommand{\setempty}{\varnothing}
 
 \newcommand{\nat}{\mathbb{N}}
 
 \newcommand{\integer}{\mathbb{Z}}

 \newcommand{\co}{\mathrm{co}\mbox{-}}

 \newcommand{\CC}{{\cal C}}
 \newcommand{\FF}{{\cal F}}

 \newcommand{\LL}{{\cal L}}

 \newcommand{\MM}{{\cal M}}
 
 \newcommand{\PP}{{\cal P}}

 \newcommand{\dl}{\mathrm{L}}
 \newcommand{\nl}{\mathrm{NL}}
 \newcommand{\p}{\mathrm{P}}
 \newcommand{\np}{\mathrm{NP}}

 \newcommand{\poly}{\mathrm{poly}}

 \newcommand{\fl}{\mathrm{FL}}

 \newcommand{\cfl}{\mathrm{CFL}}
 
 \newcommand{\dcfl}{\mathrm{DCFL}}

\theoremstyle{plain}
\theoremheaderfont{\bfseries}
 \newtheorem{theorem}{Theorem}[section]
 \newtheorem{lemma}[theorem]{Lemma}
 \newtheorem{proposition}[theorem]{{\bf Proposition}}

 {\theorembodyfont{\rmfamily}
  \newtheorem{definition}[theorem]{Definition}}
 {\theorembodyfont{\rmfamily} \newtheorem{example}[theorem]{Example}}
 {\theorembodyfont{\rmfamily} }

 \newenvironment{proofof}[1]{\vspace*{5mm} \par \noindent
         {\bf Proof of #1.\hs{2}}}{\hfill$\Box$ \vspace*{3mm}}

 \newenvironment{proof}{\par \noindent
            {\bf Proof. \hs{2}}}{\hfill$\Box$ \vspace*{3mm}}

 \newenvironment{yproof}{\par \noindent
            {\bf Proof. \hs{2}}}{\hfill$\Box$ \vspace*{3mm}}

 \newcommand{\ceilings}[1]{\lceil #1 \rceil}

\newcommand{\ignore}[1]{}

\newcommand{\track}[2]{[\:\begin{subarray}{c} #1 \\%
      #2 \end{subarray} ]}

 \newcommand{\oned}{1\mathrm{D}}

 \newcommand{\onen}{1\mathrm{N}}
 \newcommand{\twod}{2\mathrm{D}}
 
 \newcommand{\twon}{2\mathrm{N}}

 \newcommand{\para}{\mathrm{para}\mbox{-}}

 \newcommand{\phsp}{\mathrm{PHSP}}

 \newcommand{\onedpd}{\mathrm{1DPD}}
 \newcommand{\onenpd}{\mathrm{1NPD}}
 \newcommand{\twodpd}{\mathrm{2DPD}}
 \newcommand{\twonpd}{\mathrm{2NPD}}
 \newcommand{\logdcfl}{\mathrm{LOGDCFL}}
 \newcommand{\logcfl}{\mathrm{LOGCFL}}
 \newcommand{\phps}{\mathrm{PHPS}}

\begin{document}

\pagestyle{plain}
\setcounter{page}{1}


\begin{center}
{\Large {\bf Parameterizations of Logarithmic-Space Reductions, Stack-State Complexity of Nonuniform Families of Pushdown Automata, and a Road to the  LOGCFL$\subseteq$LOGDCFL/poly Question}} \bs\\
{\sc Tomoyuki Yamakami}\footnote{Present Affiliation: Faculty of Engineering, University of Fukui, 3-9-1 Bunkyo, Fukui 910-8507, Japan} \bs\\
\end{center}


\begin{abstract}
The complexity class $\logcfl$ (resp., $\logdcfl$) consists of all languages that are many-one reducible to context-free (resp., deterministic context-free) languages using logarithmic space. These complexity classes have been studied over five decades in connection to parallel computation since they are located between Nick's classes $\mathrm{NC}^1$ and $\mathrm{NC}^2$.
In contrast, the state complexity of nonuniform finite-automaton families was first discussed in the 1970s and it has been extensively explored lately for various finite-automata families. We extend this old subject to the stack-state complexity (i.e., the total number of inner states plus simultaneously pushable stack symbol series) of nonuniform families of various pushdown automata.
We introduce reasonable  ``parameterizations'' of LOGCFL and LOGDCFL and apply them as a technical tool to establish a close connection between the $\logcfl \subseteq \logdcfl/ \poly$ question and the polynomial  stack-state complexity of nonuniform families of two-way pushdown automata. We also discuss the precise computational complexity of polynomial-size one-way pushdown automata.

\vs{2}
\n{\bf Key words.} parameterized decision problem, promise problem, LOGCFL, LOGDCFL, logarithmic-space reduction, stack-state complexity, pushdown automata, polynomial-size advice
\end{abstract}


\sloppy
\section{Background and an Overview}\label{sec:introduction}

Let us quickly review necessary background materials and overview the main contribution of this work.

\subsection{Nonuniform State Complexity Classes}

We start with looking into the parallel complexity classes $\logdcfl $ and $\logcfl $, which  are the collections of all languages that are \emph{logarithmic-space many-one reducible} (or $\dl$-m-reducible) to appropriately-chosen deterministic context-free languages and context-free languages, respectively.
The class $\logcfl $ has been extensively studied since its first appearance in 1971 by Cook \cite{Coo71}. It is well-known that $\mathrm{NC}^{1}\subseteq \dl\subseteq \logdcfl \subseteq \logcfl\subseteq \mathrm{NC}^{2}$, where $\dl$ is the logarithmic-space complexity class and $\mathrm{NC}^{k}$ indicates the $k$th Nick's class.
Founded on the arguments in \cite{Gal77,Har72}, Sudborough \cite{Sud78} characterized $\logcfl $ as well as $\logdcfl $ in terms of two different machine models \emph{with no use} of $\dl$-m-reductions. One of these models is Cook's auxiliary pushdown automaton model \cite{Coo71}. We remark that, if a working hypothesis known as the \emph{linear space hypothesis}\footnote{The \emph{linear space hypothesis} (LSH) states that, for any constant $\varepsilon\in[0,1)$, a special  $\nl$-complete problem, called $\mathrm{2SAT}_3$, cannot be deterministically solved in polynomial time using $O(n^{\varepsilon})$ space \cite{Yam17a}.} \cite{Yam17a,Yam17b} is true, then $\logcfl$ is different from $\logdcfl$.
By further supplementing Karp-Lipton style advice of polynomial length to underlying $\dl$-m-reduction functions, we obtain  \emph{advised-$\dl$-m-reductions}. These advised reductions naturally induce $\logdcfl /\poly$ and $\logcfl /\poly$ respectively from $\logdcfl $ and $\logcfl $.  In Section \ref{sec:define-LOGCFL}, we will discuss two different characterizations of those advised complexity classes
$\logcfl/\poly$ and $\logdcfl/\poly$.
It is not clear at present that $\logcfl$ is included in $\logdcfl/\poly$.

Toward the $\logcfl \subseteq \logdcfl/ \poly$ question, in this work, we wish to ``parameterize'' $\logcfl/\poly$ and $\logdcfl/\poly$
by introducing a
reasonable ``parameterization'' of the aforementioned advised-$\dl$-m-reductions to define $\logcfl/\poly$ and $\logdcfl/\poly$.
Here, a \emph{parameterized decision problem}  over an alphabet $\Sigma$ refers to a pair $(L,m)$ of a language $L$ and a size parameter $m$, where $m$ is a function assigning a ``size'' to each input \cite{Yam17a,Yam17b}. A typical size parameter is the binary length $m_{bin}(x)=|x|$ of each input $x$.

To explain the goal of this work, we first review the old results of \cite{BL77,SS78} regarding  the $\nl\subseteq \dl/\poly$ question in terms of state complexity of families of two-way finite automata.
The ``size'' of a finite automaton can be measured by the number of inner states used by the automaton and this gives rise
to the notion of \emph{state complexity}.
There have been fundamental studies conducted on the state complexity of various finite automata. In the 1970s, Berman and Lingas \cite{BL77} and Sakoda and Sipser \cite{SS78} focused particularly on the families of \emph{two-way} deterministic and nondeterministic finite automata (or 2dfa's and 2nfa's, for short) of polynomial state complexities.
After a long recess since their initial works, Kapoutsis \cite{Kap09,Kap12} revitalized the study of the subject and started a systematic study on the nonuniform setting of polynomial state complexities of 2dfa's and 2nfa's. Following these works, Kapoutsis \cite{Kap14} and Kapoutsis and Pighizzini \cite{KP15} later made  significant progress, and Yamakami \cite{Yam18,Yam19a,Yam19b} further expanded the study to a wider subject.

The focal points of \cite{Kap14,KP15,SS78} were set on the nonuniform state complexity classes $\twod$ and $\twon$ of promise problems solved by nonuniform families $\{M_n\}_{n\in\nat}$ of 2dfa's and 2nfa's\footnote{Throughout this paper, we follow the formalism of \cite{Yam18,Yam19a} and fix an input alphabet $\Sigma$ over all machines $M_n$ in the same family of machines. This point is different from, e.g., \cite{Kap09,Kap12,Kap14}.}  $M_n$ having polynomial state complexity (i.e., using $n^{O(1)}$ inner states) in clear analogy with the complexity classes $\p$ and $\np$. Similarly to nonuniform circuit families, we here cope with ``nonuniform'' families of machines that take input strings of ``arbitrary'' sizes.
Berman and Lingas as well as Sakoda and Sipser discovered that a  relationship between $\mathrm{2D}$ and $\mathrm{2N}$ is closely connected to another relationship between the space-bounded complexity classes $\dl$ (deterministic logarithmic-space class) and $\nl$ (nondeterministic logarithmic-space class). Later, Kapoutsis and Pighizzini demonstrated that $\mathrm{2N}/\poly\subseteq \mathrm{2D}$ iff $\nl\subseteq \dl/\poly$, where $\dl/\poly$ is an advice version of $\dl$ and $\twon/\poly$ is a subclass of $\twon$ whose input instances given to underlying 2nfa's are restricted to strings of polynomial lengths.
This equivalence makes it possible to translate standard advised complexity classes into nonuniform state complexity classes. This phenomenon has been observed also in other nonuniform state complexity classes \cite{Yam18,Yam19a}, including classes induced by probabilistic and  quantum finite automata. We remark that an important discovery of \cite{Yam18} is the fact that nonuniform state complexity classes are more closely related to parameterized complexity classes, which naturally include standard (non-advised) complexity classes as special cases.

In sharp contrast to $\twod$ and $\twon$, \emph{one-way deterministic and nondeterministic finite automata} (or 1dfa's and 1nfa's) equipped with polynomially many inner states depict a completely different landscape. The corresponding nonuniform state complexity classes $\oned$ and $\onen$ are proven to be distinct \emph{with no assumption} (see, e.g., \cite{Kap12}).

\subsection{The LOGCFL$\subseteq$LOGDCFL/poly Question}

We have reviewed a logical equivalence between the $\twon/\poly \subseteq \twod$ question and the $\nl\subseteq \dl/\poly$ question.
So far, similar equivalences have been observed only for families of various types of ``finite automata''.  Along this line of study, this work intends to expand the scope of the study of finite automata to \emph{deterministic/nondeterministic pushdown automata}. Unlike finite automata, pushdown automata rely on both inner states and stack symbols. Those elements are crucial in describing the ``size'' of pushdown automaton because, by increasing the size of stack alphabet, we can easily reduce the number of inner states down to even $2$.
Therefore, the total number of both inner states and simultaneously pushable series of stack symbols is treated distinctively and is referred to as the \emph{stack-state complexity}
throughout this work (see Section \ref{sec:pushdown-automata} for its  precise definition).
For our convenience, we will introduce in Section \ref{sec:promise-problem} the notations $\twodpd$  and $\twonpd$ in direct analogy to $\twod$ and $\twon$, respectively, using families of two-way deterministic and nondeterministic pushdown automata having polynomial stack-state complexities. Similarly, we introduce $\onedpd$ and $\onenpd$ in Section \ref{sec:proof-oneway} based on the one-way model of pushdown automata.
By analogy to the $\nl\subseteq \dl/\poly$ question,
a direct application of ``parameterizations'' of $\logcfl $ and $\logdcfl/\poly$ establishes the following equivalence relationship:
$\twonpd/\poly \subseteq \twodpd$ iff $\logcfl\subseteq \logdcfl/\poly$.
This is a pushdown-automaton analogue of the aforementioned result of Kapoutsis and Pighizzini \cite{KP15}.
This result strongly motivates us to conduct an intensive study on  $\twonpd$ and $\twodpd$ toward answering a long-standing open question concerning the complexities of $\logcfl$ and $\logdcfl$.

As for appropriate ``parameterizations'' of $\logcfl/\poly$ and $\logdcfl/\poly$, since they are defined by an advice form of $\dl$-m-reductions to languages in $\cfl$ and $\dcfl$, we will consider a
``parameterization'' of those reduction functions.
Section \ref{sec:define-LOGCFL} will further introduce the parameterization of advised-$\dl$-m-reduction functions, from which we can naturally define $\para\logcfl/\poly$ and $\para\logdcfl/\poly$.
We will demonstrate in Section \ref{sec:proof-twoway} a close relation between the collapse of $\twonpd/\poly$ to $\twodpd$ and the collapse of a restricted form of $\para\logcfl/\poly$ down to $\para\logdcfl/\poly$.

In contrast to the two-way machine model, we will look into two nonuniform stack-state complexity classes $\onedpd$ and $\onenpd$ based on the one-way model of pushdown automata in Section \ref{sec:proof-oneway} because the one-way model is much easier to handle than the two-way model. This situation is similar to the known separation of $\oned\neq\onen$ \cite{Kap09}. We will claim the clear difference between $\onedpd$ and $\onenpd$.
We will actually show  a much stronger statement (i.e., $\onen\nsubseteq \onedpd$ and $\onedpd\nsubseteq \onen$) than this one.

\section{Foundations of The Rest of This Work}\label{sec:preparation}

We will explain the basic notions and notation that the reader needs to read through the rest of this work.

\subsection{Sets, Numbers, and Alphabets}

Given a set $A$, $\PP(A)$ denotes the \emph{power set} of $A$, i.e., the set of all subsets of $A$. The notation $\nat$ denotes the set of all \emph{natural numbers}, including $0$. We further set $\nat^{+}$ to be $\nat-\{0\}$. For two integers $m$ and $n$ with $m\leq n$, $[m,n]_{\integer}$ expresses the \emph{integer interval} $\{m,m+1,\ldots,n\}$, opposed to real intervals.
In particular, when $n\geq1$, $[1,n]_{\integer}$ is abbreviated as $[n]$.
All \emph{logarithms} are taken to the base $2$ and all \emph{polynomials} are assumed to have nonnegative integer coefficients.

An \emph{alphabet} is a nonempty finite set of ``symbols'' or ``letters.'' A \emph{string} over alphabet $\Sigma$ is a finite sequence of symbols in $\Sigma$ and its \emph{length} is the total number of symbols in the string. We use the notation $|x|$ for the \emph{length} of string $x$. The \emph{empty string} is a unique string of length $0$ and is denoted by $\lambda$.
Given an alphabet $\Sigma$, the notation $\Sigma^n$ (resp., $\Sigma^{\leq n}$) denotes the set of all strings over $\Sigma$ of length exactly $n$ (resp., at most $n$).  The notation  $\Sigma^*$ indicates the union $\bigcup_{n\in\nat}\Sigma^n$. Given a string $x$ and an index $i\in[|x|]$, we write $x_{(i)}$ for the $i$th symbol of $x$. A \emph{language} over alphabet $\Sigma$ is a subset of $\Sigma^*$ and its \emph{complement} is $\Sigma^*- L$, which is succinctly denoted by $\overline{L}$.

To express a compound pair of strings, we use the \emph{track notation} of \cite{TYL10}. Given two alphabets $\Sigma$ and $\Theta$, the notation $[\Sigma,\Theta]$ denotes a new alphabet consisting of all symbols of the form  $\track{\sigma}{\tau}$ for $\sigma\in\Sigma$ and $\tau\in\Theta$. We write each string over this new alphabet as $\track{x}{y}$ for $x\in\Sigma^n$ and $y\in\Theta^n$, where $n\in\nat^{+}$.
For notational convenience, we expand this notation to two strings $x$ and $y$ of different lengths, using a special symbol $\#$ not in $\Sigma\cup\Theta$, as follows: if $|x|<|y|$, then $\track{x}{y}$ expresses $\track{x\#^{m}}{y}$ with $m=|y|-|x|$, and if $|x|>|y|$, then $\track{x}{y}$ indicates $\track{x}{y\#^{m}}$ with $m=|x|-|y|$.
Notice that $\track{x}{y}$ is formally a string over the compound alphabet $[\Sigma,\Theta,\#]$, which is defined to be  the set $[\Sigma,\Theta]\cup\{\track{\#}{\tau},\track{\sigma}{\#}\mid \sigma\in\Sigma,\tau\in\Gamma\}$.

A function $f:\Sigma^*\to\Sigma^*$ (resp., $f:\nat\to\Sigma^*$) is said to be \emph{polynomially bounded} if there exists a polynomial $p$ for which $|f(x)|\leq p(|x|)$ (resp., $|f(n)|\leq p(n)$) holds for all $x\in\Sigma^*$ (resp., $n\in\nat$).
In contrast, $f:\nat\to\Sigma^*$ is \emph{length-preserving} if $|f(n)|=n$ holds for all $n\in\nat$. A function $f:\Sigma^*\to\nat$ is called \emph{polynomially honest} if there is a polynomial $p$ satisfying $|x|\leq p(f(x))$ for any $x\in\Sigma^*$.

\subsection{FL and FL/poly}\label{sec:FL-FLpoly}

A \emph{Turing machine} considered in this work is equipped with a read-only input tape, a rewritable work tape, and (possibly) a write-once\footnote{A tape is \emph{write-once} if its tape head never moves to the left and, whenever it writes a nonempty symbol, it must move to the right blank cell.} output tape.
For the basics of Turing machines, the reader refers to \cite{HU79} as well as \cite{Yam18,Yam19a}. Given two alphabets $\Sigma$ and $\Gamma$, a function $f:\Sigma^*\to\Gamma^*$ is in $\fl$ if there is a \emph{deterministic Turing machine} (or a DTM, for short) $M$ with a designated write-once output tape such that, given any input $x$, $M$ halts in polynomial time  and produces $f(x)$ on the output tape using only $O(\log|x|)$ work space.
It is important to note that this space bound is applied only to the work tape.
By further supplementing Karp-Lipton style ``advice'' to underlying DTMs, we can formulate an advised version of $\fl$, denoted $\fl/\poly$, by analogy with $\p/\poly$.
This can be done by providing such a DTM (briefly called an \emph{advised DTM}) with a read-only \emph{advice tape}, which carries an \emph{advice string} of length $n^{O(1)}$ over an appropriate advice alphabet $\Gamma$, where $n$ indicates any input length. Those advice strings are provided to the advised DTM by an \emph{advice function} mapping  $\nat$ to $\Gamma^*$. Such an advice function is not necessarily computable in general.

For later use, we intend to state a useful characterization of $\fl/\poly$. For the sake of completeness, we include the proof of this characterization.

\begin{lemma}\label{FL-advice}
For any function $f$, it follows that $f\in \fl/\poly$ iff  there exist a polynomially-bounded advice function $h$ and a function $g\in\fl$ such that $f(x)=g(\track{x}{h(|x|)})$ for all $x$.
\end{lemma}

\begin{yproof}
(Only If -- part) Since $f\in\fl/\poly$, we take a polynomial $p$, an  advice function $h$, and an underlying advised DTM $M$ such that (i) $|h(n)|\leq p(|x|)$ holds for all $n\in\nat$ and (ii) $M$ takes two inputs $x$ and $h(|x|)$ written on two separate tapes and eventually produces $f(x)$ on its output tape. We combine those two inputs to form a new string $\track{x}{h(|x|)}$. We want to design a new DTM, say, $N$.
To ensure the logarithmic-space bound of $N$ in the following simulation, $N$ first moves its tape head to the right and marks the rightmost tape cell of the $O(\log{n})$ work tape by leaving a special symbol $\dashv$.
The machine $N$ starts with input $\track{x}{w}$ and simulates $M$ on the input string pair $(x,w)$ given on the input and the advice tapes.
To remember the locations of two tape heads of $M$ for $x$ and $u$, we also use extra $O(\log|x|)$ work space to store the corresponding tape cell indices because of $|w|\leq p(|x|)$.
We define $g(z)$ to be the outcome of $N$ on input $z$. In particular, $f(x)$ equals $g(\track{x}{h(|x|)})$ for any $x$.

(If -- part) Conversely, assume that there are a polynomially-bounded advice function $h$ and a function $g\in\fl$ satisfying  $f(x)=g(\track{x}{h(|x|)})$ for any $x$. Take a logarithmic-space DTM $M$ that computes $g$. Consider another advised DTM $N$ that behaves as follows: on input strings $x$ and $w$ on an input and an advice tapes, simulate $M$ on $\track{x}{w}$ using an extra counter to remember the length $|x|$. This requires only extra $O(\log|x|)$ space. It is easy to check that $N$ computes $f$ correctly when $h(|x|)$ is provided as $w$.
\end{yproof}

In the original definitions of both $\logcfl$ and $\logdcfl$ discussed in Section \ref{sec:introduction}, $\fl$-functions play a key role as reduction functions. Formally, given two languages $A$ over $\Sigma$ and $B$ over $\Gamma$,  $A$ is \emph{logarithmic-space many-once reducible} (or $\dl$-m-reducible, for short) to $B$ if there exists a function $f:\Sigma^*\to\Gamma^*$ (called a \emph{reduction function}) in $\fl$ such that, for any $x\in\Sigma^*$, $x\in A$ iff $f(x)\in B$.
Given a language family $\CC$, $\mathrm{LOG}(\CC)$ denotes the collection of all languages $L$ that are $\dl$-m-reducible to some  languages in $\CC$.
In a similar way, we can define $\mathrm{LOG}/\poly(\CC)$ by replacing ``$\fl$'' in the above definition with ``$\fl/\poly$.'' In the presence of advice, we use the term of \emph{advised-$\dl$-m-reduction}.

\subsection{Pushdown Automata}\label{sec:pushdown-automata}

Context-free languages are defined by context-free grammars. Those languages are also characterized by \emph{one-way nondeterministic pushdown automata} (or  1npda's). A 1npda $M$ is formally defined as a nonuple $(Q,\Sigma,{\{\vdash,\dashv\}}, \Gamma,\delta,q_0,\bot, Q_{acc},Q_{rej})$, where $Q$ is a finite set of inner states, $\Sigma$ is an input alphabet, $\vdash$ and $\dashv$ are the left and the right endmarkers,  $\Gamma$ is a stack alphabet, $\delta: (Q-Q_{halt})\times\check{\Sigma}_{\lambda}\times\Gamma\to \PP(Q\times \Gamma^{\leq e})$ is a transition function with $\check{\Sigma} = \Sigma\cup\{\vdash,\dashv\}$, $\check{\Sigma}_{\lambda}=\check{\Sigma}\cup\{\lambda\}$,  $q_0$ is an initial state in $Q$, $\bot$ is the bottom marker in $\Gamma$, $Q_{acc}$ and $Q_{rej}$ are sets of accepting and rejecting states in $Q$, respectively, with $Q_{acc}\cap Q_{rej}=\setempty$ and $Q_{halt} = Q_{acc}\cup Q_{rej}$.
Note that $e$ is called the \emph{push size} of $M$.
If $M$ further satisfies the following \emph{deterministic requirement}, then it is called a \emph{one-way deterministic pushdown automaton} (or a 1dpda): (i) $|\delta(q,\sigma,a)|\leq1$ for any $(q,\sigma,a)\in Q\times\check{\Sigma}_{\lambda}\times\Gamma$ and (ii) whenever $\delta(q,\varepsilon,a)\neq\setempty$, it follows that $\delta(q,\sigma,a)=\setempty$ for any symbol $\sigma\in\check{\Sigma}$.
When $M$ is deterministic, we simply write $\delta(q,\sigma,a)=(p,\gamma)$ instead of $(p,\gamma)\in\delta(q,\sigma,a)$.
A \emph{stack content} refers to a series of symbols stored sequentially from the bottom to the top in a stack.
We express such a stack content as $a_1a_2\cdots a_n$, where  $a_1=\bot$ and $a_n$ is a topmost symbol. The \emph{stack height} is the length of this stack content.

A \emph{configuration} of $M$ is a triplet $(q,w,\gamma)$, where $q\in Q_n$, $w\in\Sigma^{*}$, and $\gamma\in (\Gamma-\{\bot\})^*\bot$.
This depicts a circumstance where $M$ is in inner state $q$, a tape head is scanning the leftmost symbol of $w$, and $\gamma$ is a stack content.
The \emph{initial configuration} is $(q_0,{\vdash{x}\dashv},\bot)$. A transition $(p,\tau)\in\delta(q,\sigma,a)$ indicates that, if $M$'s current configuration is of the form $(q,\sigma w,\gamma a)$, $M$ changes $q$ to $p$ and replaces $a$ by $\tau$. If $a\neq\lambda$, then $M$'s tape head must move to the right.
For two configurations $conf_1$ and $conf_2$, $conf_1\vdash conf_2$ means that $conf_2$ is obtained from $conf_1$ by a single application of $\delta$ (which corresponds one step of $M$).
If we take a finite number of steps (including zero steps), we write $conf_1\vdash^* conf_2$.

The value $|Q|+|\Gamma^{\leq e}|$ is referred to as the \emph{stack-state complexity} of $M$. This notion is compared to the \emph{state complexity}, which indicates $|Q|$, of a finite automaton.
Given any string $x$, we say that $M$ \emph{accepts} (resp., \emph{rejects}) $x$ if $M$ begins with the initial configuration,  reads ${\vdash\!{x}\!\dashv}$, enters an inner state in $Q_{acc}$ (resp., $Q_{rej}$), and halts. Given a language $L$, $M$ \emph{recognizes} $L$ if (i) for any $x\in L$, $M$ accepts $x$ and (ii) for any $x\in\overline{L}$, $M$ rejects $x$. To express this language $L$, we often use the notation $L(M)$. At this moment, we formally introduce two fundamental families $\cfl$ and $\dcfl$ as the collections of all languages recognized by 1npda's and by 1dpda's, respectively.

As for two-way versions of 1npda's and 1dpda's, which are succinctly called 2npda's and 2dpda's, we modify their aforementioned definition of $M$ as follows. A new transition function $\delta$ maps $(Q-Q_{halt})\times \check{\Sigma}_{\lambda}\times \Gamma$ to $\PP(Q\times \Gamma^{\leq e}\times D)$, where $D=\{-1,0,+1\}$.
Assume that $M$ is in inner state $q$, scanning $\sigma$ on an input tape and $a$ on a topmost stack cell. A transition of the form $(q,z,d)\in \delta(q,\sigma,a)$ causes $M$ to change $q$ to $p$, replace $a$ by $z$, and move an input-tape head in direction $d$. Note that, when $M$ reads $\lambda$, the tape head must stay still, i.e., must take the value $d=0$.

\subsection{Advice Extensions of LOGCFL and LOGDCFL}\label{sec:advice-extension}

Let us  define $\logdcfl/\poly$ and $\logcfl/\poly$, which are respectively advice versions of $\logdcfl$ and $\logcfl$, and state a few important  characterizations of them.
We first review an advice version of $\cfl$.

Karp-Lipton style advice for pushdown automata was discussed in \cite{Yam10} and the language family $\cfl/n$ was introduced there by splitting an input tape of each underlying 1npda into two separate tracks, one of which holds a standard input string $x$ and the other holds an advice string of \emph{length equal to} $x$.
For distinction, when advice is given to an underlying 1npda, we call such a machine an \emph{advised-1npda} to emphasize the use of advice. Since an advised-1npda moves its tape head only in one direction until it either reads the right endmarker or enters a halting state before the endmarker, the 1npda reads the advice string only once from left to right.
This advice model is essentially different from the one equipped with ``separate'' advice tapes whose heads
can freely move in two directions. We will discuss this two-tape model later.

Consider a compound alphabet $\Sigma_{\Theta} = [\Sigma,\Theta,\#]$ composed of two alphabets $\Sigma$ and $\Theta$.
Given a language $L$ over $\Sigma_{\Theta}$ and an advice function $h:\nat\to\Theta^*$, we define $L[h]$ as the language $\{x\mid \track{x}{h(|x|)}\in L\}$ over $\Sigma$.
By extending this notation, for a given function $f$, we write $f[h]$ to denote the function $g$ defined as $g(x) = f(\track{x}{h(|x|)})$ for all $x$.
With the help of these notations,  $\cfl/n$ is precisely composed of all languages $L[h]$ for  length-preserving advice functions $h$ and languages $L\in\cfl$.
Similarly, we define $\dcfl/n$ using a deterministic version of advised-1npda's, which are called \emph{advised-1dpda's} (\emph{with no separate advice tape}).

Now, we are ready to define $\logcfl/\poly$ and $\logdcfl/\poly$ using advised-$\dl$-m-reductions.

\begin{definition}\label{def-advice-class}
The advised complexity class $\logcfl/\poly $ (resp., $\logdcfl/\poly$) is defined to be $\mathrm{LOG}/\poly(\cfl)$ (resp., $\mathrm{LOG}/\poly(\dcfl)$).
\end{definition}

The advised families $\logcfl/\poly$ and $\logdcfl/\poly$ are quite robust classes in the following sense. We further strengthen this robustness in Lemma \ref{CFL-poly-case}.

\begin{lemma}\label{CFL-robust}
$\logcfl/\poly = \mathrm{LOG}/\poly(\cfl/n)$ and $\logdcfl/\poly = \mathrm{LOG}/\poly(\dcfl/n)$.
\end{lemma}

\begin{yproof}
We show only the first statement because the second one is similarly proven.  Since $\cfl\subseteq \cfl/n$, we instantly obtain $\logcfl/\poly \subseteq \mathrm{LOG}/\poly(\cfl/n)$.
Conversely, let $L$ denote any language over alphabet $\Sigma$ in $\mathrm{LOG}/\poly(\cfl/n)$. Take a function $f\in\fl/\poly$, a polynomial $p_1$, and a language $A\in\cfl/n$ over alphabet $\Xi$ such that $L= \{x\in\Sigma^*\mid f(x)\in A\}$
and $|f(x)|\leq p_1(|x|)$ for all $x\in\Sigma^*$.
By Lemma \ref{FL-advice}, we further take a function $g\in\fl$, a polynomial  $p_2$, and an advice function $\ell:\nat\to\Gamma^*$ for advice alphabet $\Gamma$ satisfying $f=g[\ell]$ and  $|\ell(n)| \leq p_2(n)$ for all $n\in\nat$. Note that $g$ maps $(\Sigma_{\Gamma})^*$ to $\Xi^*$, where $\hat{\Sigma}_{\Gamma} = [\Sigma,\Gamma,\#]$.
Moreover, we take a language $B\in \cfl$, a polynomial $p_3$, and an advice function $h:\nat\to\Theta^*$ for advice alphabet $\Theta$ such that $A=B[h]$ and $|h(n)|\leq p_3(n)$ for all $n\in\nat$.

Here, we abbreviate as $\hat{h}(n)$ the string $h(0)\natural h(1)\natural \cdots \natural h(n)$ with a separator $\natural$ and we intend to set  $k(n)$ to be $\track{\ell(n)}{\hat{h}(n)}$, which is a string over the compound alphabet $\hat{\Gamma}_{\Theta} = [\Gamma,\Theta\cup\{\natural\},\#]$.
Next, we define $s(\track{x}{w}) = \track{g(z)}{u_{|g(z)|}}$, where $w=\track{y}{u_0\natural u_1\natural \cdots \natural u_{|x|}}$ in $(\hat{\Gamma}_{\Theta})^*$ and
$z=\track{x}{y}$ in $(\hat{\Sigma}_{\Gamma})^*$.
Finally, let $\hat{\Xi}_{\Theta} = [\Xi,\Theta\cup\{\natural\},\#]$ and define $t=s[k]$, which is a new reduction function from $\Sigma^*$ to $(\hat{\Xi}_{\Theta})^*$.
Since $s\in\fl$ and $k$ is polynomially bounded, $t\in\fl/\poly$ follows immediately. Since $A=\{y\mid \track{y}{h(|y|)}\in B\}$, we conclude that $x\in L$ iff $t(x)=s(\track{x}{k(|x|)}) = \track{f(x)}{h(|f(x)|)}\in B$. Therefore, $L$ belongs to $\mathrm{LOG}/\poly(\cfl) = \logcfl/\poly$.
\end{yproof}

Unlike the advised-1npda's with no separate advice tape, let us  consider another model of advised-1npda that holds standard input strings and advice strings on two separate tapes and move its advice-tape head in two directions.
Formally, a language $L$ over alphabet $\Sigma$ is in $\cfl/\poly$ if there exist an advice alphabet $\Theta$, a polynomially-bounded advice function $h:\nat\to\Theta^*$, and an advised-1npda $M$ (equipped with two separate tapes) such that, for any $x\in\Sigma^*$, $x\in L$ iff $M$ starts with $x$ on an input tape and $h(|x|)$ on an advice tape and $M$ eventually accepts $x$ by moving an advice-tape head freely in two directions. Obviously, $\cfl/n\subseteq \cfl/\poly$ follows.

\begin{lemma}\label{CFL-poly-case}
$\logcfl/\poly$ and $\logdcfl/\poly$ coincide with $\mathrm{LOG}/\poly(\cfl/\poly)$ and $\mathrm{LOG}/\poly(\dcfl/\poly)$, respectively.
\end{lemma}

For technical reason, we first prove the following characterization lemma.

\begin{lemma}\label{reduction-advice-LOGCFL}
Let $L$ be any language. The following statements are logically equivalent.
\renewcommand{\labelitemi}{$\circ$}
\begin{enumerate}\vs{-2}
  \setlength{\topsep}{-2mm}%
  \setlength{\itemsep}{1mm}%
  \setlength{\parskip}{0cm}%

\item $L\in\logcfl/\poly$.

\item There exist a polynomially-bounded advice function $h$ and a language $K\in \logcfl$ such that $L=K[h]$.
\end{enumerate}\vs{-2}
The same statements also hold for $\logdcfl/\poly$.
\end{lemma}

\begin{proof}
(1 $\Rightarrow$ 2)
Assume that $L\in\logcfl/\poly$. Since $L$ is in $\mathrm{LOG}/\poly(\cfl)$, there exist a function $f\in\fl/\poly$ and a 1npda $M$ such that $L=\{x\mid \text{ $M$ accepts $f(x)$ }\}$.
By Lemma \ref{FL-advice}, there exist a function $g\in\fl$ and a polynomially-bounded function $h$ satisfying $f(x) = g(\track{x}{h(|x|)})$ for all $x$. We define $K$ as the set $\{z\mid \text{ $M$ accepts $g(z)$ }\}$. It then follows that
$L=\{x\mid \text{ $M$ accepts $g(\track{x}{h(|x|)})$ }\} = \{x\mid \track{x}{h(|x|)}\in K\} = K[h]$.

(2 $\Rightarrow$ 1)
We assume that $L=K[h]$ for a polynomially-bounded advice function $h$ and a language $K\in \logcfl$.
Take a function $f\in\fl$ and a language $A\in\cfl$ satisfying $K=\{z\mid f(z)\in A\}$. We define $g(x) = f(\track{x}{h(|x|)})$ for all $x$.
By Lemma \ref{FL-advice}, $g$ is in $\fl/\poly$. It then follows that $L=\{x\mid \track{x}{h(|x|)}\in K\} = \{x\mid g(x)\in A\}$. Thus, $L$ is in $\logcfl/\poly$.
\end{proof}

Let us return to the proof of Lemma \ref{CFL-poly-case}.

\begin{proofof}{Lemma \ref{CFL-poly-case}}
Hereafter, we intend to verify that $\logcfl/\poly = \mathrm{LOG}/\poly(\cfl/\poly)$. Notice that the deterministic case is similarly proven. Since $\cfl\subseteq \cfl/\poly$, we obtain $\logcfl/\poly \subseteq \mathrm{LOG}/\poly(\cfl/\poly)$.
For the other inclusion, let $L$ denote any language in $\mathrm{LOG}/\poly(\cfl/\poly)$.
Take an advised-$\dl$-m-reduction function $f$, a polynomially-bounded advice function $h$, and an advised-1npda $M$ such that, for any string $x$, $M$ starts with input string $x$ on an input tape and advice string $h(|x|)$ on an advice tape, and $x\in L$ holds exactly when $M$ accepts $(f(x),h(|f(x)|))$.

Since $f\in\fl/\poly$, $f$ is polynomially bounded, and thus there exists a polynomial $p_1$ satisfying $|f(x)|\leq p_1(|x|)$ for all $x$.
By Lemma \ref{FL-advice}, there exists a function $g\in\fl$ and a polynomially-bounded advice function $\ell$ satisfying $f=g[\ell]$.
Take another polynomial $p_2$ for which $\ell(n)\leq p_2(n)$ for all $n$.
In a way similar to the proof of Lemma \ref{CFL-robust}, we set $\hat{h}(n) = h(0)\natural h(1)\natural \cdots \natural h(p_2(n))$ for each $n$. We define $B$ as $\{\track{x}{u} \mid u=\track{s}{w}\wedge \text{ $M$ accepts $(g(\track{x}{s}),w_{|g(\track{x}{s})|})$ } \}$, where $w$ is of the form $w_0\natural w_1\natural \cdots \natural w_m$. Note that $B\in\logcfl$.
We define another advice function $k$ as $k(n)=\track{\ell(n)}{\hat{h}(n)}$ for any $n$. It thus follows that $\track{x}{k(|x|)}\in B$ iff $M$ accepts $(g(\track{x}{\ell(|x|)}),h(|f(x)|))$. Hence, we obtain $L=\{x\mid \track{x}{k(|x|)}\in B\} = B[k]$. By Lemma \ref{reduction-advice-LOGCFL}, this implies that $L$ is in $\logcfl/\poly$.
\end{proofof}

The characterizations given in Lemmas \ref{CFL-robust} and \ref{CFL-poly-case} leave unstated the use of another plausible complexity class $\mathrm{LOG}(\cfl/\poly)$. This is because it is unclear at present that  $\mathrm{LOG}(\cfl/\poly)$ coincides with $\logcfl/\poly$. See Section \ref{sec:discussion} for a more discussion.

Concerning $\dl/\poly$ and $\nl/\poly$, it is known that $\nl/\poly \subseteq \dl/\poly$ iff $\nl\subseteq \dl/\poly$ (see, e.g., \cite{Yam19a}). A similar equivalence also holds for $\logcfl/\poly$ and $\logcfl/\poly$, as shown in the next lemma. This fact will be used in Section \ref{sec:proof-proposition}.

\begin{lemma}\label{LOGCFL-remove}
$\logcfl/\poly \subseteq \logdcfl/\poly$ if and only if $\logcfl\subseteq \logdcfl/\poly$.
\end{lemma}

\begin{yproof}
The implication from left to right is trivial since $\logcfl$ is properly included in $\logcfl/\poly$. Conversely, assume that $\logcfl\subseteq \logdcfl/\poly$. Let $L$ denote any language over alphabet $\Sigma$ in $\logcfl/\poly$.  By Lemma \ref{reduction-advice-LOGCFL}, there are an advice alphabet $\Theta$, a polynomially-bounded advice function $h:\nat\to\Theta^*$, and a language $K\in\logcfl$ over the compound alphabet $[\Sigma,\Theta,\#]$ satisfying $L=K[h]$. For readability, we write $\hat{\Theta}_{\Sigma}$ for $[\Sigma,\Theta,\#]$.  Our assumption then yields $K\in\logdcfl/\poly$. We then apply Lemma \ref{reduction-advice-LOGCFL} for $\dcfl$ and obtain an advice alphabet $\Gamma$, a polynomially-bounded advice function $g:\nat\to\Gamma^*$, and a language $A\in\logdcfl$ over the alphabet $[\hat{\Theta}_{\Sigma},\Gamma,\#]$ satisfying $K=A[g]$. As a new advice function $k$, we set  $k(n)=\track{h(n)}{g(n)}$ for any $n\in\nat$  and define $B=\{\track{x}{z}\mid z=\track{x}{h(|x|)} \wedge \track{z}{g(|x|)}\in A\}$ by treating $z$ as a string over $\hat{\Theta}_{\Sigma}$. It then follows that $L=\{x\mid \track{x}{k(|x|)}\in B\}$. This obviously implies that $L\in\logdcfl/\poly$.
\end{yproof}

\subsection{Two-Way Auxiliary Pushdown Automata}\label{sec:two-way-aux}

With the use of reductions in $\fl/\poly$, we have introduced the advised complexity classes $\logcfl/\poly$ and $\logdcfl/\poly$ in Section \ref{sec:advice-extension}. Here, we provide another characterization of them \emph{with no use} of advised $\dl$-m-reduction. This will be quite useful in the proof of our main theorem (Theorem \ref{character-twoway}) in Section \ref{sec:proof-twoway}.
Recall that
Sudborough \cite{Sud78} characterized $\logcfl$ (as well as $\logdcfl$) in terms of Cook's auxiliary pushdown automata \cite{Coo71}. A \emph{two-way nondeterministic auxiliary pushdown automaton} (or an aux-2npda, for short) is an extension of a 2npda by attaching an additional two-way rewritable work tape.
As demonstrated in \cite{Sud78}, a language $L$ belongs to $\logcfl$ iff there exists an aux-2npda that recognizes $L$ in polynomial time using logarithmic work space.
We further expand such an aux-2npda by augmenting Karp-Lipton style advice as follows. For the sake of convenience, we call such a machine an \emph{advised-aux-2npda}. An advised-aux-2npda uses an extra read-only tape called an \emph{advice tape} on which an advice string is written. Note that all tape heads of the advised-aux-2npda can move in two directions.

Given an advised-aux-2npda $M$, an advice function $h$, and a language $L$, we say that $M$ \emph{recognizes $L$ with the help of} $h$ if (i) $M$ takes standard input $x$ and advice string $h(|x|)$ and (ii) for any $x$, if $x\in L$, then $M$ accepts, and otherwise, $M$ rejects.

The following is an advice version of Sudborough's characterization of $\logcfl$ in terms of aux-2npda's.

\begin{lemma}\label{character-LOG-CFL}
Let $L$ be any language. The following statements are logically equivalent.
\renewcommand{\labelitemi}{$\circ$}
\begin{enumerate}\vs{-2}
  \setlength{\topsep}{-2mm}%
  \setlength{\itemsep}{1mm}%
  \setlength{\parskip}{0cm}%

\item $L\in\logcfl/\poly$.

\item There exist a polynomial-time, logarithmic-space advised-aux-2npda $M$ and a polynomially-bounded advice function $h$ such that $M$ recognizes $L$ with the help of $h$.
\end{enumerate}\vs{-2}
The same statements hold for $\logdcfl/\poly$ as well.
\end{lemma}

\begin{yproof}
(1 $\Rightarrow$ 2)
This follows from \cite{Sud78}, in which $\logcfl$ is characterized by logarithmic-space aux-2npda's running in polynomial time. Let $K$ denote any language in $\logcfl/\poly$. By Lemma \ref{reduction-advice-LOGCFL}, there exist a language $K$ in $\logcfl$ and an advice function $h$ for which $L=K[h]$. Since $K\in\logcfl$, by \cite{Sud78}, there exists an aux-2npda $M$ that recognizes $K$ in polynomial time using logarithmic work space.
We thus conclude that, if $x\in L$, then $M$ accepts $\track{x}{h(|x|)}$, and otherwise, $M$ rejects $\track{x}{h(|x|)}$.
Notice that $M$ uses a single input tape, which is made up of two tracks. We split these two tracks of the input tape of $M$ into two separate tapes, one of which is an advice tape for an advice string. We denote by $N$ the obtained machine. Clearly, $N$ is an advised-aux-2npda and recognizes $L$ with the help of $h$.

(2 $\Rightarrow$ 1)
Assume that $L$ is recognized by a certain advised-aux-2npda $M$ with a polynomially-bounded advice function $h$ in polynomial time using logarithmic space.
Recall that $M$ has both an input tape and an advice tape with two separate tape heads along them
other than an auxiliary tape as well as a stack.
Since all tape heads of $M$ freely move in two directions, it is possible to treat an input tape and an advice tape as two tracks of a single input tape equipped with a single tape head.  It is important to note that this modification requires additional $O(\log|x|)$ memory bits to remember the locations of the two tape heads of the input and the advice tapes. We therefore obtain a new aux-2npda $N'$ that takes an input of the form $\track{x}{h(|x|)}$ and simulates $M$ on the pair $(x,h(|x|))$ of input strings.
We define $f(x)= \track{x}{h(|x|)}$ for any $x$ and set $K=\{z\mid \text{ $N'$ accepts $z$ }\}$. By the characterization of \cite{Sud78}, $K$ belongs to $\logcfl$. Since $L=K[h]$, Lemma \ref{reduction-advice-LOGCFL} concludes that $L$ is in $\logcfl/\poly$.
\end{yproof}

\subsection{Families of Promise Problems and Stack-State Complexity Classes}\label{sec:promise-problem}

Given an alphabet $\Sigma$, a \emph{promise decision problem} over $\Sigma$ is a pair $(A,B)$ of sets satisfying that $A,B\subseteq\Sigma^*$ and $A\cap B=\setempty$, where $A$ is viewed as a set of ``positive'' instances and $B$ represents a set of ``negative'' instances of the promise decision problem.
Naturally, we expand a single promise decision problem to a ``family'' of promise decision problems over a single alphabet.
Fix an alphabet $\Sigma$ and let $\LL=\{(L^{(+)}_n,L^{(-)}_n)\}_{n\in\nat}$ denote such a family of promise decision problems over $\Sigma$. It is important to remark that $\Sigma$ does not depend on the choice of $n$ (see, e.g., \cite{Yam18,Yam19a,Yam19b}).
All strings in $\bigcup_{n\in\nat}(L^{(+)}_n\cup L^{(-)}_n)$ are distinguished as  \emph{valid} or \emph{promised} strings.
At this moment, we demand neither $L^{(+)}_n\cap L^{(+)}_m = \setempty$ nor $L^{(-)}_n\cap L^{(-)}_m = \setempty$ for any distinct pair $m,n\in\nat$.
We say that $\LL$ has a \emph{polynomial ceiling} if there exists a polynomial $r$ satisfying $L^{(+)}_n\cup L^{(-)}_n\subseteq \Sigma^{\leq r(n)}$ for all indices $n\in\nat$.

To solve a family $\LL=\{(L^{(+)}_n,L^{(-)}_n)\}_{n\in\nat}$ of promise decision problems, we use a ``family'' of pushdown automata. Such a family is expressed as $\MM=\{M_n\}_{n\in\nat}$, where  each $M_n$ has the form $(Q_n,\Sigma,{\{\vdash,\dashv\}}, \Gamma_n,\delta_n, q_0,\bot, Q_{acc,n},Q_{rej,n})$ with $\delta_n: (Q_n-Q_{halt,n})\times \check{\Sigma}_{\lambda}\times\Gamma\to \PP(Q_n\times \Gamma^{\leq e_n})$, where $Q_{halt,n} = Q_{acc,n}\cup Q_{rej,n}$.
This machine family $\MM$ is said to \emph{solve} $\LL$ if, for any index $n\in\nat$, (i) for all $x\in L^{(+)}_n$, $M_n$ accepts $x$ and (ii) for all $x\in L^{(-)}_n$, $M_n$ rejects $x$. For any other string outside of $L^{(+)}_n\cup L^{(-)}_n$, $M_n$ may possibly neither accept nor reject it.
moreover, we do not demand any ``uniformity'' of $\MM$, that is, any existence of a fixed algorithmic procedure that generates from $1^n$ the description of each machine $M_n$.

In Section \ref{sec:introduction}, we have already discussed families of 2nfa's and 2dfa's of polynomial state complexities.
In stark contrast to the state complexity of 2nfa's and 2dfa's, we need to consider the \emph{stack-state complexity} of 2npda's and 2dpda's because we can reduce the number of inner states of pushdown automata at will by increasing their stack alphabet size.
In this work,  we are interested in families of 2npda's and 2dpda's having polynomial stack-state complexities. Analogously to the nonuniform classes $\twon$ and $\twod$, we introduce two complexity classes, $\twonpd$ and $\twodpd$, where the suffix ``PD'' stands for ``pushdown''.

\begin{definition}
The nonuniform stack-state complexity class $\twonpd$ is composed of all nonuniform families of promise decision problems solvable by appropriate families of 2npda's whose stack-state complexities are bounded from above by a fixed polynomial. Moreover, $\twonpd/\poly$ consists of all families in $\twonpd$ that have polynomial ceilings. In a similar manner, we define $\twodpd$ using 2dpda's instead of 2npda's.
\end{definition}

If we use 1npda's and 1dpda's in place of 2npda's and 2dpda's, we analogously obtain $\onenpd$ and $\onedpd$, respectively. These nonuniform complexity classes will be extensively discussed in Section \ref{sec:proof-oneway}.

\section{Parameterizations of LOGCFL/poly and LOGDCFL/poly}\label{sec:define-LOGCFL}

Toward the main goal of this work, we intend to parameterize the complexity classes $\logcfl$ and $\logdcfl$ as well as their advised  counterparts $\logcfl/\poly$ and $\logdcfl/\poly$.

\subsection{Parameterized Complexity Classes}

Similarly to $\para\dl/\poly$ and $\para\nl/\poly$ defined in \cite{Yam19a,Yam19b}, we wish to seek proper ``parameterizations'' of $\logcfl/\poly$ and $\logdcfl/\poly$, including $\logcfl$ and $\logdcfl$ as their special cases.
For readability, we will follow the basic terminology used in \cite{Yam17a,Yam17b,Yam18,Yam19a,Yam19b}. A \emph{parameterized decision problem} is a pair $(L,m)$ of a language $L$ and a size parameter $m$. We are particularly interested in size parameters computable using logarithmic space.
A \emph{log-space size parameter} $m$ is a function from $\Sigma^*$ to $\nat$ for a given alphabet $\Sigma$ for which its associated function mapping a string $x\in\Sigma^*$ to the string of the form $1^{m(x)}\in\{1\}^*$ belongs to $\fl$ \cite{Yam17a}; that is, there is a DTM $M$ (equipped with a read-only input tape, a rewritable work tape, and a write-once output tape) such that, for any string $x\in\Sigma^*$, $M$ takes $x$ as an input and produces $1^{m(x)}$ on its output tape in $|x|^{O(1)}$ time using $O(\log|x|)$ work space.

How can we define  $\para\logcfl/\poly$ and $\para\logdcfl/\poly$ in a reasonable and systematic way? Since $\logcfl/\poly$ and $\logdcfl/\poly$ are defined in terms of $\fl/\poly$-functions in Definition \ref{def-advice-class}, we first need to look for a natural parameterization of $\fl/\poly$-functions. Parameterizations of logarithmic-space computation was also discussed in, e.g., \cite{CM13,EST12}.
Given a function $f:\Sigma^*\to\Gamma^*$ for two alphabets $\Sigma$ and $\Gamma$ and a size parameter $m$, the pair $(f,m)$ belongs to $\para\fl/\poly$  if $m$ is logarithmic-space computable and there exists an advised DTM $M$, an advice function $h:\nat\to\Theta^*$ for an advice alphabet $\Theta$, and a polynomial $p$ such that $M$ takes an input string $x$ on its input tape and an advice string $h(|x|)$ on its advice tape, and $M$ produces $f(x)$ on its output tape within time $p(m(x)|x|)$ using space $O(\log{m(x)|x|})$, provided that $h$ satisfies $|h(|x|)|\leq p(m(x)|x|)$.
With the use of $\para\fl/\poly$,  we  introduce $\para\mathrm{LOG}/\poly(\CC)$ for each language family $\CC$ as the collection of all parameterized problems $(L,m)$ such that there exist a function $f$ and a language $A\in\CC$ satisfying: (i) $(f,m)\in \para\fl/\poly$ and (ii) $L=\{x \mid f(x)\in A\}$.
When we do not use advice, we naturally obtain $\para\fl$ as a special case of $\para\fl/\poly$.

\begin{definition}
The parameterized complexity classes $\para\logcfl/\poly$ and $\para\logdcfl/\poly$ are defined as $\para\mathrm{LOG}/\poly(\cfl)$ and $\para\mathrm{LOG}/\poly(\dcfl)$, respectively.
\end{definition}


Is there any a close relationship between $\para\logcfl/\poly$ and appropriately-parameterized advised-aux-2npda's?
We answer this question by proving the following characterization lemma.

\begin{lemma}\label{paraLOGCFL-character}
Let $(L,m)$ be any parameterized problem. The following three statements are logically equivalent.
\renewcommand{\labelitemi}{$\circ$}
\begin{enumerate}\vs{-2}
  \setlength{\topsep}{-2mm}%
  \setlength{\itemsep}{1mm}%
  \setlength{\parskip}{0cm}%

\item $(L,m)\in \para\logcfl/\poly$.

\item There exists an advice function $h$ and an  advised-aux-2npda $M$ such that $M$ takes any input $x$ and an advice string $h(|x|)$ of length $(m(x)|x|)^{O(1)}$ and then correctly determines whether or not $x\in L$ in time $(m(x)|x|)^{O(1)}$ using space $O(\log{(m(x)|x|)})$, where $x$ is a ``symbolic'' input.

\item There exist an aux-2npda $M$ (with no advice tape) and an advice function $h$ such that $L=L(M)[h]$, where each advice string $h(|x|)$ has length $(m(x)|x|)^{O(1)}$ and $M$ runs in time $(m(x)|x|)^{O(1)}$ using space $O(\log{(m(x)|x|)})$, where $x$ is a ``symbolic'' input.
\end{enumerate}\vs{-2}
The same statements hold for $\para\logdcfl/\poly$ and advised-aux-2dpda's.
\end{lemma}

\begin{yproof}
We note that the following argument works also for $\para\logcfl/\poly$ and advised-aux-2dpda's.

(1 $\Rightarrow$ 2) Assume that $(L,m)\in \para\logcfl/\poly$. Take a parameterized function $(f,m)\in \para\fl/\poly$ and a language $A\in\cfl$ satisfying $L=\{x\mid f(x)\in A\}$. We choose an advice function $h$ and a DTM $M$ such that $M$ takes input string $x$ on an input tape and advice string $h(|x|)$ of length $(m(x)|x|)^{O(1)}$ on an advice tape and produces $f(x)$ in time $(m(x)|x|)^{O(1)}$ using space $O(\log{(m(x)|x|)})$. We also take a 1npda $N$ recognizing $A$.

We wish to construct an aux-2npda $F$ as follows. On input $w$ of the form $\track{x}{h(|x|)}$, we simulate $N$ step by step.
For any $i\in[|w|]$, whenever $N$ tries to scan  the $i$th tape cell of $N$'s input tape, we first run $M$ on $w$ to compute the $(i+1)$th symbol of $f(x)$, and then we simulate one step of $N$'s move. The execution time of $F$ is obviously $(m(x)|x|)^{O(1)}$ since so is $M$. The space usage of $F$ is $O(\log{(m(x)|x|)})$. Since $F$ correctly solves $L$, (2) is true.

(2 $\Rightarrow$ 3) Take $(h,M)$ given in (2). We combine an input tape and an advice tape of $M$ into a single tape made up of two separate tracks.
To simulate the two tape heads of the original tapes of $M$ by a single tape head, we need to use extra $O(\log(m(x)|x|))$ bits to remember the head locations. We denote by $N$ the resulting machine.
By the definition of $N$, it follows that $L(M)$ coincides with $L(N)[h]$.

(3 $\Rightarrow$ 1) Take an advice function $h$ and an aux-2npda $M$ given in (2).
Write $K$ for $L(M)$. Note that $L=\{x\mid \track{x}{h(|x|)}\in K\}$. By Sudborough's characterization \cite{Sud78}, $K$ belongs to $\logcfl$ ($=\mathrm{LOG}(\cfl)$). Take a function $f\in\fl$ and a language $A\in\cfl$ for which $K=\{z\mid f(z)\in A\}$. We then define $g(x)=f(\track{x}{h(|x|)})$ for any $x$.
It suffices to verify that $(g,m)$ is in $\para\fl/\poly$. Since length $|h(|x|)|$ is $O(m(x)|x|)^{O(1)}$, $f(\track{x}{h(|x|)})$ can be computed in time $(m(x)|x|)^{O(1)}$ and space $O(\log{(m(x)|x|)})$. Thus, $(g,m)$ belongs to $\para\fl/\poly$.
\end{yproof}

\subsection{Polynomially-Honest Size Parameters}

A log-space size parameter $m$ is, by definition, polynomially bounded but not necessarily polynomially honest.
Here, we wish to pay special attention to polynomially-honest log-space size parameters and their associated parameterized decision problems.
For convenience, we use the notation $\phsp$ to mean the collection of all parameterized decision problems whose size parameters are polynomially honest (but not necessarily logarithmic-space computable).

Concerning the role of $\phsp$, let us demonstrate the following property, stated as Proposition  \ref{para-LOGCFL}, which will turn out to be crucial in proving a key proposition (Proposition \ref{key-prop}) in Section \ref{sec:proof-proposition}. This property bridges between $\logcfl/\poly$ (resp., $\logdcfl/\poly$) and its parameterization $\para\logcfl/\poly$ (resp., $\para\logdcfl/\poly$) if its associated log-space size parameters are restricted to be polynomially-honest.

\begin{proposition}\label{para-LOGCFL}
$\para\logcfl /\poly \cap \phsp \subseteq \para\logdcfl /\poly$ if and only if $\logcfl /\poly \subseteq \logdcfl /\poly$.
\end{proposition}

To prove this proposition, we utilize Lemmas \ref{reduction-advice-LOGCFL}, \ref{character-LOG-CFL}, and \ref{paraLOGCFL-character} as well as Sudborough's important characterizations of $\logcfl$ and $\logdcfl$ in terms of polynomial-time aux-2npda's and aux-2dpda's working with only logarithmic  space \cite{Sud78}.

\begin{proofof}{Proposition \ref{para-LOGCFL}}
(If -- part) Assume that $\logcfl/\poly \subseteq \logdcfl/\poly$. Take an arbitrary parameterized decision problem $(L,m)$ from  $\para\logcfl/\poly \cap \phsp$.
Since $(L,m)\in \phsp$, there exists a constant $k\geq1$ satisfying  $m(x)\leq |x|^k+k$ for all $x$.

By Lemma \ref{paraLOGCFL-character}, there exist an advice function $h$ and an advised-aux-2npda $N$ such that $L=L(N)[h]$ and $N$ runs in  time $(m(x)|x|)^{O(1)}\subseteq |x|^{O(1)}$ and space $O(\log{(m(x)|x|)})\subseteq O(\log|x|)$ with the help of $h$ whose length $|h(|x|)|$ is upper-bounded
by $(m(x)|x|)^{O(1)}\subseteq |x|^{O(1)}$.
Instead of giving a fixed advice string $h(|x|)$, we provide an arbitrary  string to the lower track of an single input tape of $N$ and then obtain the non-advised language $L(N)$. Write $K$ for this language $L(N)$ for simplicity.

Using Sudborough's characterization of $\logcfl$ in terms of logarithmic-space aux-2npda's \cite{Sud78}, $K$ belongs to $\logcfl$.
Since $L=K[h]$, we obtain $L \in\logcfl/\poly$ by Lemma \ref{reduction-advice-LOGCFL}. Our assumption then derives $L\in\logdcfl/\poly$. Since $m$ is polynomially honest, we conclude that $(L,m)\in\para\logdcfl/\poly$.

(Only If -- part) We start with the assumption of  $\para\logcfl/\poly\cap\phsp \subseteq \para\logdcfl/\poly$. Take any language $L$ in $\logcfl/\poly$ and set $m_{bin}(x)=|x|$ for all $x$. Notice  that $m_{bin}$ is polynomially honest. Our goal is to derive the conclusion of $L\in\logdcfl/\poly$.

By Lemma \ref{character-LOG-CFL}, there exist a polynomially-bounded function $h$ and a polynomial-time, logarithmic-space  advised-aux-2npda $N$ satisfying $L=L(N)[h]$.
Consider the parameterized decision problem $(L,m_{bin})$. Since $N$ solves $L$ with the help of $h$ in time $(m_{bin}(x)|x|)^{O(1)}$ using space $O(\log{(m_{bin}(x)|x|)})$,
$(L,m_{bin})$ belongs to $\para\logdcfl/\poly\cap \phsp$.
Our assumption thus implies that $(L,m_{bin})\in \para\logdcfl/\poly$.
Since $m_{bin}(x)=|x|$, Lemmas \ref{character-LOG-CFL} and \ref{paraLOGCFL-character} for $\logdcfl/\poly$ and $\para\logdcfl/\poly$ then conclude that  $L\in\logdcfl/\poly$.
\end{proofof}

\section{A Road to the LOGCFL$\subseteq$LOGDCFL/poly Question}\label{sec:proof-twoway}

In Sections \ref{sec:preparation}--\ref{sec:define-LOGCFL}, we have already discussed fundamental properties necessary to prove the following main theorem of this work: two characterizations of the $\logcfl\subseteq \logdcfl/\poly$ question in terms of nonuniform stack-state complexity through the parameterizations of $\logcfl/\poly$ and $\logdcfl/\poly$.

\begin{theorem}\label{character-twoway}
The following three statements are logically equivalent.
\renewcommand{\labelitemi}{$\circ$}
\begin{enumerate}\vs{-2}
  \setlength{\topsep}{-2mm}%
  \setlength{\itemsep}{1mm}%
  \setlength{\parskip}{0cm}%

\item $\mathrm{2NPD}/\poly \subseteq \mathrm{2DPD}$.

\item $\para\logcfl/\poly \cap \phps \subseteq \para\logdcfl/ \poly$.

\item $\logcfl \subseteq \logdcfl /\poly$.
\end{enumerate}
\end{theorem}

The proof of the above theorem is quite involved with the proper use of the  parameterizations of $\logcfl/\poly$ and $\logdcfl/\poly$  discussed in Section \ref{sec:define-LOGCFL}.
In Section \ref{character-twoway}, we will give the formal proof of Theorem \ref{character-twoway}, which requires a supporting statement (Proposition \ref{key-prop}) verified specifically in Section \ref{sec:proof-proposition}.

\subsection{How to Verify Theorem \ref{character-twoway}}

Our goal of the rest of this section is to prove Theorem \ref{character-twoway}. To accomplish this goal, another key statement, Proposition \ref{key-prop}, is needed. To describe the proposition, nonetheless, we need to explain important terminology from \cite{Yam19a}.


Given a promise decision problem $(L,m)$ over alphabet $\Sigma$, we define $L^{(+)}_n=L\cap \Sigma^{(n)}$ and $L^{(-)}_n = \overline{L}\cap\Sigma^{(n)}$, where $\Sigma^{(n)}$ denotes $\{x\in\Sigma^*\mid m(x)=n\}$.
This special notation $\Sigma^{(n)}$ is intentionally used here to  differentiate it from $\Sigma^n$ ($=\{x\in\Sigma^*\mid |x|=n\}$). We further define $\LL=\{(L^{(+)}_n,L^{(-)}_n)\}_{n\in\nat}$. This family $\LL$ is said to be \emph{induced from} $(L,m)$.

On the contrary, given two families $\LL=\{(L^{(+)}_n,L^{(-)}_n)\}_{n\in\nat}$ and $\hat{\LL}=\{(\hat{L}^{(+)}_n,\hat{L}^{(-)}_n)\}_{n\in\nat}$ of promise decision problems, $\hat{\LL}$ is called an \emph{extension} of $\LL$ if $L^{(+)}_n\subseteq \hat{L}^{(+)}_n$ and $L^{(-)}_n\subseteq \hat{L}^{(-)}_n$ for all $n\in\nat$. When the set $\{1^n\# x\mid n\in\nat, x\in L^{(+)}_n\cup L^{(-)}_n\}$ belongs to $\dl$, $\LL$ is said to be \emph{$\dl$-good}. Finally, a collection $\FF$ of families of promise problems is \emph{$\dl$-good} if every element in $\FF$ has an $\dl$-good extension in $\FF$.
Given an $\dl$-good family $\LL=\{(L^{(+)}_n,L^{(-)}_n)\}_{n\in\nat}$ of promise decision problems over alphabet $\Sigma$, we define $K^{(+)}_n=\{1^n\# x\mid x\in L^{(+)}_n\}$ and $K^{(-)}_n = \{1^n\# x\mid x\notin L^{(+)}_n\} \cup \{z\# x\mid z\in \Sigma^n-\{1^n\},x\in (\Sigma_{\#})^*\}\cup \Sigma^n$, where $\Sigma_{\#} = \Sigma\cup\{\#\}$. We further define $K=\bigcup_{n\in\nat} K^{(+)}_n$, from which $\overline{K}=\bigcup_{n\in\nat}K^{(-)}_n$
immediately follows.  We then define $m: (\Sigma_{\#})^*\to \nat$ by setting $m(w)=n$ if $w=1^n\# x$ for a certain string $x\in L^{(+)}_n\cup L^{(-)}_n$, and $m(w)=|w|$ otherwise. The obtained parameterized problem $(K,m)$ is said to be \emph{induced from} $\LL$.

\begin{example}\label{L-good}
As a concrete example, we intend to demonstrate that $\twonpd$ is $\dl$-good. Take an arbitrary family $\LL=\{(L_n^{(+)},L_n^{(-)})\}_{n\in\nat}$ of promise problems over alphabet $\Sigma$ in $\twonpd$. Consider a family $\MM=\{M_n\}_{n\in\nat}$ of polynomial-size 2npda's that solves $\LL$. Let $\Sigma^{(n)}= L_n^{(+)}\cup L_n^{(-)}$ for any $n\in\nat$.
Although $M_n$ is not even required to halt on inputs outside of $\Sigma^{(n)}$, we can easily modify $M_n$ to halt on all of its computation paths on any input.
We then define $\hat{L}_n^{(+)} = \{x\in\Sigma^*\mid \text{ $M_n$ accepts $x$ }\}$ and $\hat{L}_n^{(-)} = \{x\in\Sigma^*\mid \text{ $M_n$ rejects $x$ }\}$. We then set $\hat{\LL} =\{(\hat{L}_n^{(+)},\hat{L}_n^{(-)})\}_{n\in\nat}$. Since $\hat{L}_n^{(+)}\cup \hat{L}_n^{(-)} = \Sigma^*$, $L_n^{(+)}\subseteq \hat{L}_n^{(+)}$, and  $L_n^{(-)}\subseteq \hat{L}_n^{(-)}$, it follows that $\hat{\LL}$ is an extension of $\LL$. We then set $A=\{1^n\# x\mid n\in\nat, x\in \hat{L}_n^{(+)}\cup \hat{L}_n^{(-)}\}$, which equals $\{1^n\# x\mid n\in\nat, x\in\Sigma^*\}$. Clearly, this set $A$ is recognized deterministically using only logarithmic space, and thus $A$ belongs to $\dl$.
\end{example}


The following statement gives the final piece of our argument by further bridging between stack-state complexity classes and their associated parameterized classes.

\begin{proposition}\label{key-prop}
Let $L$ and $K$ be any two languages and let $m$ be any log-space size parameter. Let $\LL=\{(L^{(+)}_n,L^{(-)}_n)\}_{n\in\nat}$ be any family of promise decision problems.
\renewcommand{\labelitemi}{$\circ$}
\begin{enumerate}\vs{-1}
  \setlength{\topsep}{-2mm}%
  \setlength{\itemsep}{1mm}%
  \setlength{\parskip}{0cm}%

\item If $\LL$ is induced from $(L,m)$, then $(L,m)\in\para\logcfl/\poly\cap \phsp$ iff $\LL\in\twonpd/\poly$.

\item If $\LL$ is $\dl$-good and $(K,m)$ is induced from $\LL$, $(K,m)\in\para\logcfl/\poly\cap \phsp$ iff $\LL\in \twonpd/\poly$.
\end{enumerate}\vs{-2}
The same statements hold for the pair of $\logdcfl/\poly$ and $\mathrm{2DPD}/\poly$.
\end{proposition}

Proposition \ref{key-prop} is the most challenging statement to verify in this work.
Theorem \ref{character-twoway} is easily obtained from
this proposition as follows with an additional help of Lemmas \ref{LOGCFL-remove} and \ref{paraLOGCFL-character} and Proposition \ref{para-LOGCFL}.
Recall that Proposition \ref{para-LOGCFL} connects both $\logcfl/\poly$ and $\logdcfl/\poly$ to their natural parameterizations.

\begin{proofof}{Theorem \ref{character-twoway}}
(2 $\Leftrightarrow$ 3) This is a direct consequence of Proposition \ref{para-LOGCFL} and Lemma \ref{LOGCFL-remove}.

(1 $\Rightarrow$ 2)
Assuming that $\twonpd/\poly \subseteq \twodpd$, we wish to verify that $\para\logcfl\poly\cap \phsp \subseteq \para\logdcfl /\poly$.
Take any promise decision problem $(L,m)$ in $\para\logcfl/\poly\cap\phsp$.
Consider the family $\LL$ induced from $(L,m)$.
Proposition \ref{key-prop}(1) implies that $\LL\in\twonpd/\poly$.
By our assumption, it follows that $\LL$ is in $\twodpd/\poly$. Again applying Proposition \ref{key-prop}(1) for $\logdcfl$ and $\twodpd$, we conclude that  $(L,m)\in\para\logdcfl/\poly$.

(2 $\Rightarrow$ 1)
Conversely, we assume that $\logcfl\subseteq \logdcfl/\poly$.
Given any family $\LL$ in $\twonpd/\poly$, take an  $\dl$-good extension $\LL'$ of $\LL$ in $\twonpd/\poly$ since $\twonpd$ is $\dl$-good by Example \ref{L-good}.
Let $(K,m)$ denote the parameterized decision problem induced from $\LL'$.
By Proposition \ref{key-prop}(2), we obtain $(K,m)\in\para\logcfl/\poly \cap\phsp$. Our assumption then implies that $(K,m)\in\para\logdcfl/\poly\cap \phsp$. Proposition \ref{key-prop}(2) for $\logdcfl$ and $\twodpd$ leads to the conclusion that $\LL\in\twodpd/\poly\subseteq \twodpd$.
\end{proofof}

Theorem \ref{character-twoway} motivates us to further study the power and limitation of $\twodpd$ and $\twonpd$ in connection to the $\logcfl\subseteq \logdcfl/\poly$ question.
The proof of Proposition \ref{key-prop} will be given in the subsequent subsection.

\subsection{How to Prove Proposition \ref{key-prop}}\label{sec:proof-proposition}

Through this subsection, we intend to prove Proposition \ref{key-prop}. In the following argument, $L$ and $K$ denote two languages, $m$ denotes a log-space size parameter, and  $\LL=\{(L^{(+)}_n,L^{(-)}_n)\}_{n\in\nat}$ denotes a family of promise decision problems.
Our proof proceeds as follows.

(1) Assume that  $\LL=\{(L^{(+)}_n,L^{(-)}_n)\}_{n\in\nat}$ is induced from $(L,m)$. This implies that $L^{(+)}_n=L\cap\Sigma^{(n)}$ and $L^{(-)}_n=\overline{L}\cap\Sigma^{(n)}$ for any index $n\in\nat$, where $\Sigma^{(n)} =\{x\in\Sigma^*\mid m(x)=n\}$.

Let us begin our argument with assuming that $(L,m)\in\para\logcfl/\poly\cap\phsp$.
We want to derive the conclusion of $\LL\in\twonpd/\poly$.
Take a fixed polynomial $p$  satisfying $m(x)\leq p(|x|)$ for all $x$. Notice that $\Sigma^{(m(x))}\subseteq \Sigma^{\leq p(|x|)}$ for any $x$.
By Lemma \ref{paraLOGCFL-character}, we take an advice function $h$ and an advised-aux-2npda $N = (Q,\Sigma,{\{\vdash,\dashv\}},\Theta,\Gamma,\delta,q_0,\bot,Q_{acc},Q_{rej})$  solving $(L,m)$ in time $(m(x)|x|)^{O(1)}\subseteq |x|^{O(1)}$ using space $O(\log{(m(x)|x|)})\subseteq O(\log|x|)$ with advice strings $h(|x|)$ of size $(m(x)|x|)^{O(1)}$, where $\Theta$ is an alphabet for an auxiliary work tape. Let $e$ denote the push size of $N$.
We wish to define a family $\MM=\{M_n\}_{n\in\nat}$ of polynomial-size 2npda's for $\LL$ since this concludes that $\LL\in\twonpd$.

Fix $n\in\nat$ arbitrarily and let us define the desired machine $M_n$. For convenience, we prepare a pair $(l,w)$, which expresses a content $w$ of an auxiliary work tape and a location $l$ of its tape head.
The $M_n$'s inner states are of the form $(q,l,w)$ for $q\in Q$, $l\in[0,k\log{n}]_{\integer}$, and $w\in\Gamma^{k\log{n}}$ for a sufficiently large constant $k\geq1$.
On input $x\in\Sigma^{(n)}$ and advice string $z=h(|x|)$, $M_n$ tries to simulate $N$ on $(x,z)$ in such a way that, whenever $N$ changes $q$ to $q'$ and $(l,w)$ to $(l',w')$ in a single step, $M_n$ changes its inner state from $(q,l,w)$ to $(q',l',w')$ accordingly. Since $|Q\times [0,k\log{n}]_{\integer}\times \Gamma^{k\log{n}}| = n^{O(1)}$, $M_n$ has polynomially many inner states together with a constant size of $\Gamma^{\leq e}$, and thus $\MM$ has polynomial stack-state complexity.
Therefore, $\LL$ belongs to $\twonpd$. Since $\Sigma^{(n)} = \Sigma^{(m(x))} \subseteq \Sigma^{\leq p(|x|)}$,
we further conclude that $\LL\in\twonpd/\poly$.

Conversely, we assume that $\LL\in\twonpd/\poly$. Our goal is to obtain $(L,m)\in\para\logcfl/\poly\cap\phsp$.
Let us take a family $\MM=\{M_n\}_{n\in\nat}$ of polynomial-size 2npda's solving $\LL$. For each index $n\in\nat$, assume that $M_n$ has the form $(Q_n,\Sigma,{\{\vdash,\dashv\}},\Gamma_n, \delta_n,\bot,q_0,Q_{n,acc},Q_{n,rej})$ with push size $e_n$.
Note that there exist two polynomials $p$ and $s$ such that $L_n^{(+)}\cup L_n^{(-)} \subseteq \Sigma^{\leq s(n)}$ and $|Q_n|\leq p(n)$ for all $n\in\nat$. This yields  $|x|\leq s(m(x))$ for all $x$ since $\Sigma^{(n)} = L_n^{(+)}\cup L_n^{(-)}$.
As a result, $(L,m)\in\phsp$ follows.
Next, we wish to simulate all 2npda's $M_n$ using an appropriate  advised-aux-2npda, say, $N$ and an appropriate advice function $h$ to ensure that $(L,m)$ is indeed in $\para\logcfl/\poly$.

Since each $M_n$ may use completely different sets $Q_n$ and $\Gamma_n$, we first need to ``identify'' $Q_n$ and $\Gamma_n$ as sets of strings; i.e., $Q_n=\{0,1\}^{k(n)}$ and $\Gamma_n =\{0,1\}^{l(n)}$ for polynomials $k(n)$ and $l(n)$, where $q_0=0^{k(n)}$ and $\bot=0^{l(n)}$.
This helps us switch $\Gamma_n$ to a new stack alphabet $\Gamma=\{0,1,\#,\bot\}$ of fixed size. Take another polynomial $r$ satisfying  $2^{l(n)(e_n+1)}\leq r(n)$ for all $n\in\nat$.

We modify $M_n$ to design a new machine $M'_n$ as follows. For readability, we intentionally use multiple auxiliary work tapes for $M'_n$ because it is possible to combine those tapes into a single one with extra $O(\log{|x|})$ bits to remember the locations of work-tape heads.
As for an inner state $q\in Q_n$, we translate it into the corresponding $k(n)$-bit string $w(q)\in \{0,1\}^{k(n)}$. We use the first auxiliary work tape of $O(\log{n})$ space to store $q$ in the form of $w(q)$.
Next, we translate a stack content $\bot a_1\cdots a_m$ into $\bot\#w(a_1)\#w(a_2)\cdots \#w(a_m)$, where each symbol $a_i$ is translated  into its corresponding $w(a_i)$ of $l(n)$ symbols together with the designated separator $\#$.
One step of the simulation of $M_n$ goes as follows. We scan a stack to remove a topmost block of symbols $w(a)$ into the 2nd work tape.
We scan both the 1st and the 2nd work tapes to recover $(q,a)$ and apply a transition of $\delta_n$.
Whenever $M_n$ replaces symbol $a$ by $z=b_1b_2\cdots b_{m'} \in (\Gamma_n)^{\leq e_n}$ in a single step, we change it into $\# w(b_1)\# w(b_2)\cdots \# w(b_{m'})$.
If the 2nd work tape does not become empty, then we push a series of newly written symbols back into the stack starting with $\#$.
This newly obtained machine $M'_n$ uses a transition function of the form $\delta'_n: (\{0,1\}^{k(n)}-Q_{halt,n}) \times \check{\Sigma}_{\lambda}\times \Gamma\times \Theta_1\times \Theta_2\to \PP(\{0,1\}^{k(n)}\times \Gamma^{\leq \ceilings{\log{r(n)}}}\times \Theta_1\times\Theta_2)$, where $\Theta_1$ and $\Theta_2$ are appropriate alphabets used for the 1st and the 2nd auxiliary work tapes.

Furthermore, we translate $\delta'_n$ into a \emph{transition table}, in which each row is indexed by $(q,\sigma,a,\tau_1,\tau_2)\in \{0,1\}^{k(n)}\times \check{\Sigma}_{\lambda}\times\Gamma\times \Theta_1\times \Theta_2$ and each column enumerates all values in $\delta'_n(q,\sigma,a,\tau_1,\tau_2)$ in a predetermined order. We take an appropriate encoding of this transition table into a single string $h(n)$ of length proportional to $k(n)+r(n) = n^{O(1)}$ so that we can recover necessary transitions
at any time simply by scanning this advice string $h(n)$ from left to right by a two-way tape head.
This provides an advised-aux-2npda {}
that simulates $M'_{m(x)}$ on input $x$ in time $(m(x)|x|)^{O(1)}$ and space $O(\log{(m(x)|x|)})$ with the help of $h$. In the end, we obtain $(L,m)\in\para\logcfl/\poly$.

(2) We start with an $\dl$-good  family $\LL=\{(L^{(+)}_n,L^{(-)}_n)\}_{n\in\nat}$ of promise decision problems and $(K,m)$ denotes the parameterized decision problem with a log-space size parameter $m$, which is induced from $\LL$.

Assuming that $(K,m)\in\para\logcfl/\poly\cap \phsp$, we aim at showing that $\LL\in\twonpd/\poly$.
Lemma \ref{paraLOGCFL-character} implies the existence of an advice function $h$ and an advised-aux-2npda $N$ that together solve $(L,m)$ in time $(m(x)|x|)^{O(1)}$ and space $O(\log(m(x)|x|))$ using advice strings $h(|x|)$ of size $(m(x)|x|)^{O(1)}$.
Moreover, we assume that an auxiliary work tape of $N$ has space at most $k\log{(m(x)|x|)}+k$ for a fixed constant $k>0$ on all inputs $x$. Since $(K,m)$ is induced from $\LL$, it follows that $m(w)=n$ if $w=1^n\# x$ and $x\in L^{(+)}_n\cup L^{(-)}_n$, and $m(w)=|w|$ otherwise.
Since $m$ is polynomially honest, we take a polynomial $p$ satisfying $|x|\leq p(m(x))$ for all $x$. Therefore, if $x\in L_n^{(+)}\cup L_n^{(-)}$, then $|1^n\# x|\leq p'(n)$ holds, where $p'(n) = n+p(n)+1$.

Let us define a family $\MM=\{M_n\}_{n\in\nat}$ of 2npda's for $\LL$ as follows. To cope with the auxiliary work tape of $N$, we treat a pair $(l,w)$ of a tape content $w$ and a tape head location $l$ as a part of extra inner states of $M_n$.
Each advice string $h(|w|)$ is embedded into the inside of $M_n$ as part of inner states.
Since $(K,m)\in\phsp$, $x\in L^{(+)}_n\cup L^{(-)}_n$ implies $|x|\leq p(m(x)) = p(n)$. On input $x$, $M_n$ reads $x$ and simulates $N$ on $(w, h(|w|))$ with $w=1^n\# x$.
This is possible because $h(|1^n\# x|)$ consists of at most $p(n)$ symbols if $x\in L_n^{(+)}\cup L_n^{(-)}$. If $N$ changes $(l,w)$ at any step, then $M_n$ changes its associated inner states. This shows that $\MM$ solves $\LL$ correctly. Thus, $\LL\in\twonpd$ follows.
Since $|1^n\# x|\leq p'(n)$ for all $x\in L_n^{(+)}\cup L_n^{(-)}$, we further conclude that  $\LL\in\twonpd/\poly$.

To show the converse, we assume that $\LL\in\twonpd/\poly$. There is a polynomial $t$ such that, for any $n$ and $x$, $x\in L^{(+)}_n\cup L^{(-)}_n$ implies $|x|\leq t(n)$.  By the definition of $m$, we obtain $|x|\leq t(m(1^n\# x))$. Thus, $m$ is polynomially honest.
Take a family $\MM=\{M_n\}_{n\in\nat}$ of polynomial-size 2npda's solving $\LL$.
Our goal is to show that $(K,m)\in\para\logcfl/\poly$ by constructing an appropriate advised-aux-2npda, say, $N$.
In a way similar to (1), for each index $i\in\nat$,
we encode all transitions of $M_i$ into a single string, say, $\alpha_i$.
Let $r(l)= \max\{m(w)\mid |w|=l\}$. We then define an advice string $h(l) = \alpha_0\# \alpha_1\# \cdots \# \alpha_{r(l)}$. The desired machine $N$ works as follows. On input $(w,h(|w|))$ with $w=1^n\# x$, if $|x|\leq t(n)$, then $N$ simulates $M_n$ on $x$ by following the transitions encoded into $\alpha_n$ using $O(\log|w|)$-space auxiliary work tape to keep track of the changes of inner states of $M_n$; otherwise, $N$ rejects the input. Clearly, $N$ solves $(K,m)$ with the help of $h$. Therefore, $(K,m)$ belongs to $\para\logcfl/\poly$.

\s
This completes the proof of the proposition.

\section{Case of One-Way Models}\label{sec:proof-oneway}

We have discussed the two nonuniform stack-state complexity classes $\twonpd$ and $\twodpd$ in Section \ref{sec:proof-twoway}.
Unfortunately, it still remains open whether or not $\twonpd$ coincides with $\twodpd$ and $\twonpd/\poly$ is included in $\twodpd$. However, when we move away from the two-way model of polynomial-size families of pushdown automata
to the \emph{one-way model} (that is, 1npda's and 1dpda's), it is actually possible to prove the clear difference between determinism and nondeterminism.

We begin with formally defining $\onedpd$ and $\onenpd$. The class $\onenpd$ is composed of all families of promise problems $\LL=\{(L^{(+)}_n,L^{(-)}_n)\}_{n\in\nat}$ such that there exist families of 1npda's $\MM=\{M_n\}_{n\in\nat}$ for which each $M_n$ solves $(L^{(+)}_n,L^{(-)}_n)$ for any $n\in\nat$. The deterministic counterpart of $\onenpd$ is denoted by $\onedpd$.
With these notations, we claim that $\onedpd$ and $\onenpd$ are indeed different. In the following two propositions, we actually show that $\onen$ and $\onedpd$ are \emph{incomparable}. Since $\oned\subseteq \onedpd$ and $\onen\subseteq \onenpd$, this incomparability leads to the desired separation $\onedpd\neq \onenpd$.

\begin{proposition}\label{oneDPD-vs-oneN}
$\onedpd\nsubseteq \onen$.
\end{proposition}

\begin{proof}
We define ${\cal PAL}=\{(Pal^{(+)}_n,Pal^{(-)}_n)\}_{n\in\nat}$ over the ternary alphabet $\Sigma=\{0,1,\#\}$ by setting $Pal^{(+)}_n=\{w\# w^R\mid w\in\{0,1\}^n\}$ (even-length marked palindromes) and $Pal^{(-)}_n = \Sigma^{2n+1}- Pal^{(+)}_n$.
In what follows, we claim that  (1) ${\cal PAL}\in\onedpd$ and (2) ${\cal PAL}\notin \onen$.

(1) Fix $n\in\nat$ arbitrarily. Assuming $x=u\# v$ for two strings $u,v\in\{0,1\}^*$, we first push $u$ into a stack with the help of the separator $\#$ and then check that $v$ matches $u^R$ by popping $u$ from the stack in the reverse order. At the same time, we check that  $|u|= n$ by using polynomially many inner states. It is possible to check whether $x$ is of the form $u\# v$ for $u,v\in\{0,1\}^*$ without using any stack. Thus, the entire procedure can be implemented on an appropriate 1dpda. As a consequence, ${\cal PAL}$ belongs to  $\onedpd$.

(2) The following argument loosely follows \cite{Kap12} by way of contradiction. Assume that there is a family $\MM=\{M_n\}_{n\in\nat}$ of 1nfa's having polynomial state complexity that solves ${\cal PAL}$. Let $p$ denote an appropriate polynomial such that $M_n$ has at most $p(n)$ inner states for any $n\in\nat$. Fix $n$ arbitrarily and consider the $\{0,1\}^{\leq n}\times \{0,1\}^{\leq n}$ matrix $A$ whose indices $(u,v)$ are of the form with $u,v\in\{0,1\}^{\leq n}$ and their entries are $1$ if $u\# v\in Pal^{(+)}_n$, and $0$ otherwise. This matrix $A$ is clearly a diagonal matrix (according to an appropriate index ordering). For each index $u\in\{0,1\}^{\leq n}$, we fix an accepting computation path of $M_n$ on $u\# u$ and denote by $q_{u}$ an inner state that $M_n$ takes while $M_n$'s tape head is crossing the boundary between $u\#$ and $u$.

Since there are $2^{n+1}$ diagonal entries in $A$ and $M_n$ has at most $p(n)$ inner states, there must be two distinct indices  $u_1,u_2\in\{0,1\}^{\leq n}$ for which $q_{u_1}=q_{u_2}$. Let us consider the input string $u_1\# u_2^R$. Since $q_{u_1}=q_{u_2}$, $M_n$ must accepts $u_1\# u_2^R$, implying $u_1=u_2$. However, this contradictions $u_1\neq u_2$. Therefore, ${\cal PAL}$ is outside of $\onen$.
\end{proof}

\begin{proposition}\label{separation-oneway}
$\onen\nsubseteq \onedpd$.
\end{proposition}


\begin{yproof}
It is easy to verify that $\co\onedpd = \onedpd$ by exchanging between  accepting states and rejecting states of underlying 1dpda's. Thus, it suffices to show that $\co\onen\nsubseteq \onedpd$.

Our example family of promise problems is ${\cal DUP}  =\{(Dup^{(+)}_{n},Dup^{(-)}_{n})\}_{n\in\nat}$ (duplication) over the binary alphabet $\Sigma=\{0,1\}$, where
$Dup^{(+)}_{n}=\{w w\mid w\in \Sigma^{n}\}$
and $Dup^{(-)}_{n} = \Sigma^{2n} - Dup^{(+)}_{n}$.
To separate $\co\onen$ from $\onedpd$, we wish to show that (1) ${\cal DUP}$ belongs to $\co\onen$ and (2) ${\cal DUP}$ is not in $\onedpd$.

(1) Given each promise problem  $(Dup^{(+)}_n,Dup^{(-)}_n)$, we design a  machine $M_n$ as follows. On input $x$ of the form $uv$ with $u,v\in\{0,1\}^n$, \emph{universally} (i.e., co-nondeterministically)  choose indices $i\in[|u|]$ and check that $u_{(i)}=v_{(i)}$. Splitting $x$ into $u$ and $v$ is possible by $M_n$ using polynomially many inner states. This machine universally accepts the input $x$ exactly when $x$ is in $Dup^{(+)}_{n}$. Therefore, ${\cal DUP}\in\co\onen$ follows.

(2) We show this assertion by way of contradiction. Assume that ${\cal DUP}$ is in $\onedpd$ and take a family $\MM=\{M_n\}_{n\in\nat}$ of 1dpda's  that solves ${\cal DUP}$ with a set $Q_n$ of inner states. Furthermore, we choose a polynomial $p$ satisfying $|Q_n|\leq p(n)$ for any $n\in\nat$. It is possible to assume that, without loss of generality, that $M_n$ empties its stack at the end of computation and that there is only one accepting state, say, $q_{acc}$.

Let us consider configurations of each machine $M_n$.  We define the set $A_n$ to be $\{(x,q_1,y,q_2,z)\mid xyz\in Dup^{(+)}_n, \exists \gamma [(q_0,{\vdash{xyz}\dashv},\bot)\vdash^* (q_1,{yz\dashv},\gamma) \vdash^* (q_2,{z\dashv},\gamma) \vdash^* (q_{acc},\lambda,\bot)] \}$, provided that $M_n$'s stack height does not go below $|\gamma|$ while reading $y$. Given $(q_1,q_2,\ell_1,\ell_2)$ with $q_1,q_2\in Q_n$ and $\ell_1,\ell_2\in[0,n+1]_{\integer}$, we set $B_{q_1,q_2,\ell_1,\ell_2}
=\{(x,q_1,y,q_2,z) \in A_n \mid |x|=\ell_1, |y|=\ell_2\}$.
We claim that (*) there exists a quadruple $(q_1,q_2,\ell_1,\ell_2)$ satisfying $|B_{q_1,q_2,\ell_1,\ell_2}|\geq 2$. If this is true, then we take such a quadruple $(q_1,q_2,\ell_1,\ell_2)$ and two distinct elements $(x_1,q_1,y_1,q_2,z_1)$ and $(x_2,q_1,y_2,q_2,z_2)$ from $B_{q_1,q_2,\ell_1,\ell_2}$. By the definition, $M_n$ must accept the inputs $x_1y_1z_1$ and $x_2y_2z_2$. Because of $B_{q_1,q_2,\ell_1,\ell_2}$, $x_1y_2z_1$ and $x_2y_1z_2$ are also accepted by $M_n$, implying $x_1y_2z_1,x_2y_1z_2\in Dup_n^{(+)}$. However, this is impossible.
As a consequence, we obtain ${\cal DUP}\notin\onedpd$.

We still need to prove the pending claim (*). Assume otherwise; namely, $|B_{q_1,q_2,\ell_1,\ell_2}|\leq 1$ for any quadruple $(q_1,q_2,\ell_1,\ell_2)$. We define a partial function $g$ by setting $g(q_1,q_2,\ell_1,\ell_2) = (x,y,z)$ if $(x,q_1,y,q_2,z)\in B_{q_1,q_2,\ell_1,\ell_2}$, and $g(q_1,q_2,\ell_1,\ell_2)$ is undefined otherwise. Note that the size of the domain $dom(g)$ is at most $|Q_n|^2(n+1)^2$. This implies that $|\bigcup_{q_1,q_2,\ell_1,\ell_2}B_{q_1,q_2,\ell_1,\ell_2}|\leq p(n)^2(n+1)^2< 2^n$ for a sufficiently large $n$. This contradicts the fact that $|\bigcup_{q_1,q_2,\ell_1,\ell_2}B_{q_1,q_2,\ell_1,\ell_2}|\geq 2^n$  because $g$ must produce all triplets $(x,y,z)$ satisfying $xyz\in Dup^{(+)}_n$.
\end{yproof}

\section{A Brief Discussion and Future Directions}\label{sec:discussion}

Throughout this work, we have expanded to pushdown automata the scope of ``nonuniform state complexity'' classes based on finite automata, such as $\twod$ and $\twon$, which were initiated in the 1970s by Berman and Lingas \cite{BL77} and Sakoda and Sipser \cite{SS78}.
We have introduced the notion of nonuniform stack-state complexity and have defined two important complexity classes $\twodpd$ and $\twonpd$ using 1dpda's and 2npda's, respectively.

As a main theorem, we have established an exact relationship between the $\logcfl\subseteq \logdcfl/\poly$ question and the $\twonpd/\poly\subseteq \twodpd$ question by way of introducing a reasonable ``parameterization'' of $\logcfl$ and $\logdcfl/\poly$. This relationship ensures the importance of the study of nonuniform  stack-state complexity classes in automata theory. Using the one-way machine models of 1dpda's and 1npda's, we have defined two more natural complexity classes $\onedpd$ and $\onenpd$. Unlike the case of two-way models, we have shown that $\onedpd\neq \onenpd$.

For the sake of the avid reader, we intend to raise a few open problems associated with the results of this work.

\renewcommand{\labelitemi}{$\circ$}
\begin{enumerate}\vs{-2}
  \setlength{\topsep}{-2mm}%
  \setlength{\itemsep}{1mm}%
  \setlength{\parskip}{0cm}%

\item The most important question left open in this work is the $\twonpd/\poly \subseteq \twodpd$ question. At this moment, we speculate that $\twonpd/\poly$ is not included in $\twodpd$, implying $\logcfl\nsubseteq \logdcfl/\poly$ by Theorem \ref{character-twoway}, but there seems no apparent evidence that strongly supports our speculation. Therefore, it would be desirable to find such an evidence to ensure the correctness of our speculation. Another relevant open question is whether or not $\twodpd\neq\twonpd$ holds.

\item Numerous nonuniform state complexity classes based on the one-way machine models are discussed in \cite{Yam19a} but not all relationships among these classes are determined. We wonder how  $\onedpd$ and $\onenpd$ fit into the landscape of all such complexity classes. Can we prove clear separations of $\onedpd$ and $\onenpd$ from those complexity classes?

\item Lately, the language families, denoted by $\mathrm{LOG}k\mathrm{SDA}$ for all indices $k\geq2$, based on deterministic depth-$k$ storage automata were  introduced in \cite{Yam21} to extend $\logdcfl$. It is proven in \cite{Yam21} that $\logdcfl\subseteq \mathrm{LOG}k\mathrm{SDA} \subseteq \mathrm{SC}^k$ for any $k\geq2$, where $\mathrm{SC}^k$ is the $k$th Steve's class. Is it possible to expand our main result (Theorem \ref{character-twoway}) to these intriguing classes?

\item We have defined $\logcfl/\poly$ as $\mathrm{LOG}/\poly(\cfl)$ in Section \ref{sec:advice-extension}. However, we have left open the question of whether $\mathrm{LOG}(\cfl/\poly)$ coincides with $\logcfl/\poly$. If this is not the case, what is the precise complexity of  $\mathrm{LOG}(\cfl/\poly)$?

\item We still lack for a general theory of both nonuniform state complexity classes and nonuniform stack-state complexity classes. It is imperative to develop such a theory for the immediate benefit of promoting the basic understanding of the behaviors of automata in general.
\end{enumerate}

\let\oldbibliography\thebibliography
\renewcommand{\thebibliography}[1]{%
  \oldbibliography{#1}%
  \setlength{\itemsep}{-1pt}%
}
\bibliographystyle{alpha}

\end{document}